\documentclass[11pt,letterpaper, american]{article}
\usepackage[letterpaper,margin=1in]{geometry}
\usepackage[dvipsnames]{xcolor}
\usepackage{enumitem}

\usepackage{tikz} 
\usetikzlibrary{decorations.markings}
\usetikzlibrary{arrows.meta,positioning}
\tikzset{
    pe/.style={
decoration={
            markings, mark=at position 0.8 with {\fill[red] (0, 0.5pt) -- ++ (-0.25, 0.075) -- ++ (0, -0.075) -- cycle;
            }
        },
        postaction=decorate
    },
}

\newcommand\myshade{85}
\colorlet{mylinkcolor}{violet}
\colorlet{mycitecolor}{YellowOrange}
\colorlet{myurlcolor}{Aquamarine}

\usepackage{comment}
\usepackage{complexity}
\usepackage{amsmath}
\usepackage{amsthm}
\usepackage{amssymb}
\usepackage{silence}
\usepackage{thm-restate}

\usepackage[unicode=true, bookmarks=false, breaklinks=false,pdfborder={0 0 1},backref=section,colorlinks=false] {hyperref}
\hypersetup{
  linkcolor  = mylinkcolor!\myshade!black,
  citecolor  = mycitecolor!\myshade!black,
  urlcolor   = myurlcolor!\myshade!black,
  colorlinks = true,
}
\usepackage[nameinlink,english]{cleveref}

\newcommand{\ie}{, i.e., }

\def\renewtheorem#1{\expandafter\let\csname#1\endcsname\relax
  \expandafter\let\csname c@#1\endcsname\relax
  \gdef\renewtheorem@envname{#1}
  \renewtheorem@secpar
}

\def\renewtheorem@secpar{\@ifnextchar[{\renewtheorem@numberedlike}{\renewtheorem@nonumberedlike}}
\def\renewtheorem@numberedlike[#1]#2{\newtheorem{\renewtheorem@envname}[#1]{#2}}
\def\renewtheorem@nonumberedlike#1{  
\def\renewtheorem@caption{#1}
\edef\renewtheorem@nowithin{\noexpand\newtheorem{\renewtheorem@envname}{\renewtheorem@caption}}
\renewtheorem@thirdpar
}
\def\renewtheorem@thirdpar{\@ifnextchar[{\renewtheorem@within}{\renewtheorem@nowithin}}
\def\renewtheorem@within[#1]{\renewtheorem@nowithin[#1]}

\usepackage[nameinlink]{cleveref}
\renewcommand{\ref}{\cref}
\AtBeginDocument{\let\ref\cref}
\crefname{lem}{Lemma}{Lemmas}
\crefname{problem}{Problem}{Problems}
\crefname{step}{Step}{Steps}
\crefname{thm}{Theorem}{Theorems}
\crefname{proposition}{Proposition}{Propositions}
\crefname{prop}{Proposition}{Propositions}
\crefname{defn}{Definition}{Definitions}
\crefname{fact}{Fact}{Facts}
\crefname{fig}{Figure}{Figures}
\crefname{figure}{Figure}{Figures}
\crefname{section}{Section}{Sections}
\crefname{cor}{Corollary}{Corollaries}
\crefname{remark}{Remark}{Remarks}

\crefname{conjecture}{Conjecture}{Conjectures}
\crefname{alg}{Algorithm}{Algorithms}
\crefname{algorithm}{Algorithm}{Algorithms}
\crefname{ex}{Example}{Examples}
\crefname{exa}{Example}{Examples}
\crefname{example}{Example}{Examples}
\crefname{eq}{Equation}{Equations}
\crefname{equation}{Equation}{Equations}
\crefname{subequation}{equation}{equations}

\usepackage{mathptmx}
 \newcommand{\goravname}{Gorav Jindal}
\newcommand{\goravgmail}{gorav.jindal@gmail.com}

\newcommand{\goravaffil}{Max Planck Institute for Software Systems, Saarbr{\"u}cken, Germany. Part of the work was done while the author was a member of Graduiertenkolleg `Facets of Complexity/Facetten der Komplexit\"at' (GRK 2434) and Institut f\"ur Mathematik, Technische Universit\"at Berlin}

\newcommand{\petername}{Peter B{\"u}rgisser}

\newcommand{\peteruni}{pbuerg@math.tu-berlin.de}
\newcommand{\peteraffil}{Institut f{\"u}r Mathematik, Technische Universit{\"a}t Berlin, Berlin, Germany}
 \numberwithin{equation}{section}
\numberwithin{figure}{section}

\theoremstyle{plain}
\newtheorem{thm}{\protect\theoremname}[section]

\theoremstyle{definition}
\newtheorem{problem}{\protect\problemname}[section]

\theoremstyle{plain}
\newtheorem{prop}{\protect\propositionname}[section]

\theoremstyle{definition}

\theoremstyle{plain}
\newtheorem{lem}{\protect\lemmaname}[section]

\theoremstyle{plain}
\newtheorem{cor}{\protect\corollaryname}[section]

\theoremstyle{plain}
\newtheorem{conjecture}{\protect\conjecturename}[section]

\theoremstyle{definition}
\newtheorem{remark}{\protect\remarkname}[section]

\newlist{casenv}{enumerate}{4}
\setlist[casenv]{leftmargin=*,align=left,widest={iiii}}
\setlist[casenv,1]{label={{\itshape\ \casename} \arabic*.},ref=\arabic*}
\setlist[casenv,2]{label={{\itshape\ \casename} \roman*.},ref=\roman*}
\setlist[casenv,3]{label={{\itshape\ \casename\ \alph*.}},ref=\alph*}
\setlist[casenv,4]{label={{\itshape\ \casename} \arabic*.},ref=\arabic*}
\providecommand{\casename}{Case}
\providecommand{\conjecturename}{Conjecture}
\providecommand{\corollaryname}{Corollary}
\providecommand{\definitionname}{Definition}
\providecommand{\lemmaname}{Lemma}
\providecommand{\problemname}{Problem}
\providecommand{\propositionname}{Proposition}
\providecommand{\theoremname}{Theorem}
\providecommand{\remarkname}{Remark}

\usepackage{booktabs}
\newcommand{\splitatcommas}[1]{\begingroup
  \begingroup\lccode`~=`, \lowercase{\endgroup
    \edef~{\mathchar\the\mathcode`, \penalty0 \noexpand\hspace{0pt plus 1em}}}\mathcode`,="8000 #1\endgroup
}

\begin{document}
\clearpage{}\global\long\def\R{\mathbb{R}}

\global\long\def\N{\mathbb{N}}

\global\long\def\F{\mathbb{F}}

\global\long\def\C{\mathbb{C}}

\global\long\def\CC{\mathbb{C}}

\global\long\def\Q{\mathbb{Q}}

\global\long\def\Z{\mathbb{Z}}

\global\long\def\card#1{\left\vert #1\right\vert }

\global\long\def\negi{\operatorname{negative}}

\global\long\def\posi{\operatorname{positive}}

\global\long\def\flr#1{\lfloor#1\rfloor}

\global\long\def\paren#1{\left(#1\right)}

\global\long\def\brac#1{\left\{#1\right\}}

\global\long\def\cli#1{\lceil#1\rceil}

\global\long\def\lst#1#2{\splitatcommas{#1_{1},\dots,#1_{#2}}}

\global\long\def\lstl#1#2{#1_{1}\geq#1_{2}\geq\dots\geq#1_{#2}}

\global\long\def\sgn{\operatorname{sgn}}

\global\long\def\degslp{\operatorname{DegSLP}}

\global\long\def\crr{\operatorname{CountRealRoots}}

\newcommand{\twodiff}{simple}

\global\long\def\ssr{\operatorname{SSR}}

\global\long\def\lcm{\operatorname{lcm}}

\global\long\def\rad{\operatorname{rad}}

\global\long\def\lc{\operatorname{lc}}\global\long\def\cont{\operatorname{Cont}}

\global\long\def\pols{\operatorname{PolySAT}}

\global\long\def\cyc{\operatorname{Cyc}}

\global\long\def\chr{\operatorname{char}}

\global\long\def\posslp{\operatorname{PosSLP}}

\global\long\def\eqdef{:=}

\global\long\def\abs#1{\left|#1\right|}

\global\long\def\odd#1{\operatorname{Odd}(#1)}

\global\long\def\len{\operatorname{length}}
\title{On the Hardness of PosSLP}
\author{\petername\thanks{\peteraffil.\, Email: \texttt{\peteruni}} \and \goravname \thanks{\goravaffil.\, Email: \texttt{\goravgmail}}}
\date{}
\maketitle
\begin{abstract}

The problem $\posslp$ involves determining whether an integer computed by a given straight-line program is positive. This problem has attracted considerable attention within the field of computational complexity as it provides a complete characterization of the complexity associated with numerical computation. However, non-trivial lower bounds for $\posslp$ remain unknown. In this paper, we demonstrate that  $\posslp\in\BPP$ would imply that $\NP \subseteq \BPP$, under the assumption of a conjecture concerning the complexity of the radical of a polynomial proposed by Dutta, Saxena, and Sinhababu (STOC'2018).
Our proof builds upon the established $\NP$-hardness of determining if a univariate polynomial computed by an SLP has
a real root, as demonstrated by Perrucci and Sabia (JDA'2005).

Therefore, our lower bound for $\posslp$ represents a significant advancement in understanding the complexity of this problem. It constitutes the first non-trivial lower bound for $\posslp$, albeit conditionally. Additionally, we show that counting the real roots of an integer univariate polynomial, given as input by a  straight-line program, is $\#\P$-hard.

\end{abstract}

\section{Introduction}

\subsection{Straight-line Programs}

Given an integer $a$ as input, how do we decide whether $a$ is positive
or negative? This question seems very innocuous at the first glance.
Indeed, if $a$ is given as a bit string, the question is trivial.
This question becomes interesting when we are given a compact expression
for $a$ instead of its bit string representation. One such compact
way to represent an integer is by an arithmetic circuit 
or, equivalently, a straight-line program.  
These are fundamental concepts studied in algebraic complexity theory:
we refer the reader to excellent surveys~\cite{ShpilkaY10,saptharishi2021survey}.

An  {\em arithmetic circuit}
is a directed acyclic graph, whose leaves are labeled by formal variables
$\lst xn$ or scalars from the underlying field~$\F$. The non-leaf
nodes are arithmetic gates. We assume that there is a unique output
node and use the term gate for node interchangeably. Every gate of
such a circuit computes a multivariate polynomial in the canonical
way. The polynomial computed at the output gate is said to be the
polynomial computed by the circuit. The size of the circuit is defined
as the number of gates in it. 
We shall restrict our attention to constant free arithmetic circuits,
which compute univariate polynomials.
We define a {\em straight-line program~$P$, SLP} for short, to be a sequence 
of univariate integer polynomials $(a_{0},\lst a{\ell})$
such that $a_{0}=1,a_{1}=x$
and $a_{i}=a_{j}\circ a_{k}$ for all $2\leq i\leq\ell$, where $\circ\in\{+,-,*\}$
and $j,k<i$. We say that $P$ computes the univariate polynomial
$a_{\ell}$ and that $P$ has length $\ell$. Note that in this definition,
we do not allow any constants different from $1$. For an integer univariate polynomial $f\in\Z[x]$, we define $\tau(f)$ as 
the length of the smallest SLP which computes $f$. It is clear that
any SLP of length $\ell$ can be described using $O(\ell\log\ell)$
bits. Since an integer is a special case of a univariate polynomial,
SLPs can also compute integers. 

The problem $\degslp$, introduced in~\cite{Allender06onthe}, is the problem  
of computing the degree of a polynomial given as input
by an arithmetic circuit. 

\subsection{PosSLP and Related Work}\label{subsec:slppodslpdef}

Now we formally define the {\em problem $\posslp$}, 
introduced in \cite{Allender06onthe},
which is the central object of study of this paper

\begin{problem}[$\posslp$]Given an SLP $P$ computing an integer $n_{P}$, decide if $n_{P}>0$. 
\end{problem}

This problem was introduced to establish a connection between classical models of computation and 
the {\em Blum-Shub-Smale model}~\cite{SmaleRealCompu1997}. 
The latter is an extensively studied model for 
studying computations with real numbers. 
A BSS machine~$M$ runs according to a finite program and 
takes as input a finite sequence of real numbers of arbitrary length\ie 
an element of $\cup_{n}\R^{n}$. 
Moreover, $M$ has an infinite tape consisting of cells containing 
real numbers or blanks. In each step, $M$ can copy the content of one cell into another,
perform an arithmetic operation $\circ\in\{+,-,\times,\div\}$
on two cells, or branch by comparing any cell to~$0$. 
The class~$\P_{\R}^{0}$
denotes the set of decision problems decided by polynomial time 
by {\em constant free} BSS machines.
To compare this with classical complexity classes, defined via Turing machines,
one considers the Boolean part 
$\BP(\P_{\R}^{0})\eqdef\{L\cap\{0,1\}^{n}\mid L\in\P_{\R}\}$.

In \cite{Allender06onthe}, it was shown that 
the computational power of this complexity class is given by polynomial time 
computations with oracle calls to $\posslp$. That is:

\begin{prop}\label{prop:posslporaclebppr}$\P^{\posslp}=$ $\BP(\P_{\R}^{0})$. 
\end{prop}

\cite{Allender06onthe} also explained 
the relevance of $\posslp$ for numerical computation 
in a more direct way, without referring to the formal model of BSS machines,
as follows.
For any nonzero real number $r$, we can write $r=s2^{m}$ with $\frac{1}{2}\leq\abs s<1$
and $m\in\Z$. A floating point approximation of $r$ with $k$ significant
bits is a floating point number $t2^{m}$ such that $\abs{s-t}\leq2^{-(k+1)}$.

\begin{problem}[The generic task of numerical computation]
\label{prob:gtnc}Given an arithmetic circuit $C$ computing a polynomial
$f(\lst xn)$, given floating point numbers $\lst an$, and an
integer $k$~in unary, along with a promise that $f(\lst an)$ is
nonzero, compute a floating point approximation of the value of the
output $f(\lst an)$ with $k$ significant bits.
\end{problem}

\cite{Allender06onthe} showed that \ref{prob:gtnc} is polynomial
time Turing equivalent to $\posslp$. This assertion and \ref{prop:posslporaclebppr}
support the hypothesis that $\posslp$ does not have efficient algorithms.
In addition, it is not hard to see that the above defined problem $\degslp$ reduces to $\posslp$. 
This further suggests the computational intractability of $\posslp$, 
given the belief that an efficient algorithm for $\degslp$ is unlikely to exist.

As for upper bounds on $\posslp$, the following result 
in terms of the counting hierarchy~\cite{wagner186,counthierjuha2009,complexityarora2009}
is the best known result.

\begin{thm}[\cite{Allender06onthe}]
\label{thm:posslpupperbound}
$\posslp\in\P^{\PP^{\PP^{\PP}}}$.
\end{thm}

Jindal and Saranurak \cite{saranurak2012subtraction} observed that if monotone SLP complexity $\tau_+$  and SLP complexity $\tau$ of positive integers are polynomially equivalent, then $\posslp \in \Sigma_2^{\P} \subseteq \PH$. There are
several other important problems which reduce to $\posslp$.
One such well-studied and important problem is the following.

\begin{problem}[Sum-of-square-roots problem, $\ssr$]
Given a list $(\lst an)$ of positive integers and a positive integer
$k$, decide if $\sum_{i\in[n]}\sqrt{a_{i}}\geq k$.
\end{problem}

This is asked as an open problem in \cite{openssr1976}. It has connections to the Euclidean traveling salesman problem. The Euclidean
traveling salesman problem is not known to be in $\NP$, but is readily
seen to be in $\NP$ relative to an $\ssr$ oracle. 
By using classical Newton iteration, $\ssr$ reduces to $\posslp$ \cite{Allender06onthe}. 
The sum-of-square-roots problem was conjectured
to be in $\P$ in~\cite{Malajovich2001AnEV}. 

Another important problem is to decide the inequality of succinctly
represented integers \cite{succintintegers2014}. More precisely,
consider the following problem.
\begin{problem}[Inequality testing of succinctly represented integers]
\label{prob:succineq}Given positive integers $\lst an,\lst bn,\lst cm,\lst dm$,
decide if $\prod_{i=1}^{n}a_{i}^{b_{i}}\geq\prod_{i=1}^{m}c_{i}^{d_{i}}.$
\end{problem}

This problem is easily seen to be a special case of $\posslp$.
It was shown in \cite{succintintegers2014} that \ref{prob:succineq}
can be solved in deterministic polynomial time if one can prove strong
lower bounds on integer linear combinations of logarithms of positive
integers, known as the 
Lang-Waldschmidt conjecture~\cite{lang2013elliptic,opendiiph2004}, 
see~\cite[Conjecture 3.2 ]{succintintegers2014}. 
However, \ref{thm:posslpupperbound} still provides 
the best unconditional upper bound for~\ref{prob:succineq}. 

\subsection{Our Results}

Despite the non-trivial, but rather inefficient upper bound of~\ref{thm:posslpupperbound},
no hardness results are known for $\posslp$. In this
paper, we show $\posslp \in \P$ would have dramatic consequences 
for complexity theory,
assuming a variant of the radical conjecture proposed in \cite{dssjacm22}.

We define the {\em radical} $\rad(f)$ of a nonzero integer polynomial
$f\in\Z[\lst x n]$ as the product of the irreducible integer polynomials dividing~$f$. 
It is also called the {\em square-free part} of~$f$.
Note that $\rad(f)$ is uniquely defined up to a sign. 
This generalizes the radical of a nonzero integer~$n$, 
which is defined as the product of the distinct prime numbers dividing $n$.

We shall crucially rely on the following conjecture, which is a constructive variant
of the radical conjecture proposed in \cite{dssjacm22}. 
See \ref{sec:compelxityofradicals} for a discussion of~\ref{conj:constrdicalconj}.

\begin{restatable}[Constructive univariate radical conjecture]{conjecture}{constructiveradicalconjecture}
\label{conj:constrdicalconj} 
For any polynomial 
$f\in\Z[x]$, we have $\tau(\rad(f))\leq\poly(\tau(f)$).
Moreover, there is a randomized polynomial time algorithm which, 
given an SLP of size $s$ computing~$f$, 
constructs an SLP for $\rad(f)$ of size $\poly(s)$ with success probability at least 
$1-\frac{1}{\Omega(s^{1+\epsilon})}$ for some $\epsilon>0$.
\end{restatable}

In fact, for our purposes, it is enough to know that \ref{conj:constrdicalconj} 
applies to some nonzero integer multiple of $\rad(f)$. The following result is the main contribution of this paper; 
see \ref{sec:nphardnessposslp} for the proof. 

\begin{restatable}{thm}{mainresultnphardposslp}
\label{thm:mainnphardposslp}
If \ref{conj:constrdicalconj} is true and $\posslp \in \BPP$ then $\NP \subseteq \BPP$.
\end{restatable}

We also show that counting the real roots of univariate polynomials
computed by straight-line programs is $\#\P$-hard. 
For a univariate polynomial $F$, let us denote by $Z_{\R}(F)$ 
the number of its real roots, counted with multiplicity. 
Consider the problem:

\begin{problem}[$\crr$]
\label{prob:countrealroots}Given a SLP $P$ computing a univariate
polynomial $f_{P}(x)\in\Z[x]$, compute $Z_{\R}(f_{P})$. 
\end{problem}

We prove the following result in \ref{sec:realrootcountinghardness}.

\begin{restatable}{thm}{countrealrootssharppahrd}\label{thm:countrealrootshard}$\crr$
is $\#\P$-hard. \end{restatable}

\subsection{Proof Ideas}

We rely on the proof of $\NP$-hardness due to \cite{PERRUCCI2007471} of
the following problem. 
The idea for this reduction goes back to \cite{PLAISTED1984125}, 
see \ref{sec:tchebyshevandrealroots} for a detailed discussion.

\begin{problem}
\label{prob:slprealrootexist}
Given an SLP $P$ computing a univariate integer polynomial, decide whether this polynomial has a real root. 
\end{problem}

\paragraph{Proof idea for \ref{thm:mainnphardposslp}:}

The reduction in~\cite{PERRUCCI2007471} computes for a given 3SAT formula $W$ in polynomial time 
an SLP computing a univariate polynomial
$P\in\Z[x]$ such that $W$ is satisfiable iff $P$ has a real root.
All the real roots of $P$  are in
the interval $[-1,1]$. Moreover, 
every real root of $P$ (if any) has
multiplicity two, and $P$  never attains negative values. Now
we call on~\ref{conj:constrdicalconj} to construct an SLP 
computing the radical $R\eqdef\rad(P)$. 
By the definition of the radical, the real roots of~$R$
are exactly the real roots of  $P$, but with multiplicity one.
If $W$ is not satisfiable, then $P$ has no real roots, and therefore, neither does~$R$.
In this case, $R$ does not change signs on $[-1,1]$: either it completely remains
below the $x$-axis or completely remains above the $x$-axis. 
On the other hand, if $W$~is satisfiable, then $R$ does cross the $x$-axis 
at the real roots of  $P$, because every root of $R$ has multiplicity exactly one. 
Hence it attains both negative and positive
values. Therefore:
\begin{itemize}
\item If $W$ is not satisfiable, then $R$ does not change sign on the interval $[-1,1]$.
\item If $W$ is satisfiable, then $R$ attains both negative and positive values on $[-1,1]$. 
\end{itemize}
Now we sample a random rational point $a$ from the interval $[-1,1]$.
By oracle calls to $\posslp$, we can compute the sign of $R$ at $a$ and $1$. 
By construction, if $W$ is not satisfiable, then $R(a)$ has the same sign as $R(1)$. 
But if $W$ is satisfiable, then there
are choices of $a$ for which $R(a)$ and $R(1)$ have different signs. 
Our novel contribution is that we can adapt the $\NP$-hardness
reduction in~\cite{PERRUCCI2007471} such that 
$R(a)$ and $R(1)$ have different signs with a constant probability,
for a random $a\in[-1,1]$, provided $W$ is satisfiable.
To this end, we will assume that $W$ has at most one satisfying
assignment, which can be achieved using the randomized polynomial time
reduction of 3SAT to Unique-SAT \cite{VALIANT198685}. 
Under this assumption, the set of real roots of $R$ has a simpler structure, 
which allows us to prove that $R(a)$ and $R(1)$ have different signs
with constant probability for a random $a\in[-1,1]$, 
see \ref{sec:nphardnessposslp} for details.

\paragraph{Proof idea for \ref{thm:countrealrootshard}:}
 
We use the ideas developed in Section 4 of \cite{countingcurves1993}
and the reduction of \cite{PERRUCCI2007471} outlined above. 
Let us denote by $\#W$ the number of satisfying
assignments of a 3SAT formula $W$. 
It is well known that computing $\#W$ is $\#\P$-complete.
The strategy is to prove that  
$\#W$ can be computed in polynomial time,  
if oracle calls to $\crr$ are allowed.

For a given 3SAT formula $W$, 
we compute the polynomial $P$ from above, which has a real
root iff $W$ is satisfiable. 
It turns out that 
$P$ has $N(\phi)$ many roots for each satisfying
assignment~$\phi$ of $W$, where 
$N(\phi) \eqdef \prod_{p_{i}\not\in\phi}(p_{i}-1)$. 
Here $\phi$ is seen as a subset of a set of $n$ odd 
primes $\lst pn$, see \ref{sec:tchebyshevandrealroots}. In the reduction of \cite{PERRUCCI2007471}, one 
can choose any odd primes $p_i$. Here we first choose an odd prime $q$, and
then the prime $p_{i}$ is chosen from the arithmetic progression
$\{aq+2\mid a\in\N\}$. 
This can be done efficiently because primes in arithmetic 
progressions have sufficiently high density~\cite{primesinapbennett2018}. 
This implies that $N(\phi)\equiv1\bmod q$. 
We show that as a consequence, 
$\#W\equiv \frac12 Z_{\R}(P)\bmod q$. 
Therefore, by using oracle calls to $\crr$, we can compute  
$\#W\bmod q$ for any odd prime~$q$. 
Doing so for sufficiently many odd primes $q$, and using Chinese remaindering, 
we can finally compute $\#W$. 
We refer to \ref{sec:realrootcountinghardness} for details. 

\section{Preliminaries\label{sec:tchebyshevandrealroots}}

In this section we mainly recall the reduction from \cite{PERRUCCI2007471,PLAISTED1984125} 
to prove that \ref{prob:slprealrootexist} is $\NP$-hard. 
For a positive integer $n$ we write $[n]\eqdef\{1,2,\ldots,n\}$.

Let us first note the following folklore result for later use.
It follows from the observation that the graph of a real polynomial crosses the
$x$-axis only on the roots of odd multiplicity.

\begin{lem}\label{lem:rootcountparity}
Let $a,b\in\R$ with $a<b$ and $f$
is a real univariate polynomial. Assume $f(a)f(b)\neq0$. Then $f$
has an even number of real roots (counted with multiplicity) in $(a,b)$
if and only if $f(a)f(b)>0$. \qed
\end{lem}

We also note the following easy observation for later use.

\begin{lem}\label{lem:signatrational} 
Consider the following problem: given a univariate polynomial $F(x)$ as an SLP 
and $p,q\in\Z$, compute the sign of $F(p/q)$. 
This problem reduces to $\posslp$ under polynomial time many one reductions. 
\end{lem}
\begin{proof}
First notice that we have SLPs of length $O(\log p)$ and $O(\log q)$,
which compute $p$ and $q$, respectively. 
Suppose $F(x)$ is given by an SLP $P$ of length $\ell$.  
By induction on the length of SLPs, we show that we can efficiently construct from~$P$  
two SLPs $P_{1}$, $P_{2}$ of length $O(\log p+\log q+\ell)$, computing
integers $r,s$ respectively, such that $F(p/q)=r/s$.
Moreover $\sgn(r/s)=\sgn(rs)$.
\end{proof}

\subsection{Chebychev polynomials}

The Chebychev polynomials $T_{k}$ are univariate polynomials in one variable $x$ 
defined by 
$T_{0}(x) \eqdef 1$, $T_{1}(x)\eqdef x$ and for an integer $k\ge 2$ 
by the recursion 
$$
 T_{k}(x) \eqdef2xT_{k-1}(x)-T_{k-2}(x) .
$$
Clearly, the $T_{k}$ are integer polynomials. 
They have the following well known properties; see \cite{kincaid91}.

\begin{enumerate}
\item $\deg(T_{k})=k$  and the leading coefficient of $T_{k}$ is $2^{k-1}$.
\item $T_{k}(x)=\cos(k\arccos(x))$ for all $x\in[-1,1]$. 
\item The roots of $T_{k}$ are $\brac{\cos\paren{t\frac{\pi}{2k}}\mid t\in\{1,3,\ldots,2k-1\}}$. 
\item For every $p,q\in\N$, we have $T_{p}\circ T_{q}=T_{pq}$, where $\circ$ denotes the composition.
\end{enumerate}

The following is an easy consequence of the properties of Chebychev polynomials.

\begin{lem}[Lemma 1 in \cite{PERRUCCI2007471}]
\label{lem:TpqSLP} The Chebychev polynomial $T_{k}$ can
be computed by an SLP of length $O(k)$. Moreover, if $k=pq$ for 
$p,q\in\N$, then $T_{k}=T_{pq}$ can be computed
by a straight-line program of length $O(p+q)$. 
\end{lem}

\subsection{Real Roots of Univariate Polynomials and Straight-line Programs}
\label{subsec:RR}

We recall here the reduction~\cite{PERRUCCI2007471,PLAISTED1984125} 
from the well-known $\NP$-complete problem 3SAT to \ref{prob:slprealrootexist}.  
The idea is to associate with the 
$n$ literals $\lst xn$ of a 3SAT formula 
once and for all $n$ distinct odd primes $\lst pn$.
It will be convenient to abbreviate $p_{\max}\eqdef\max_i p_i $ and $p_{\min}\eqdef\min_i p_i$.
We put $M\eqdef\prod_{i\in[n]}p_{i}$ and 
enumerate the roots of 
the $M^{\textrm{th}}$ Chebychev polynomial~$T_{M}$
by the odd integers 
$t\in \odd M\eqdef\{1,3,\ldots,2M-1\}$.
Thus we define  
\begin{equation}\label{eq:def_r_M}
 r_{M}(t)\eqdef\cos\paren{t\frac{\pi}{2M}} 
 \end{equation}
and denote by 
$R_{M}\eqdef\{r_{M}(t)\mid t\in\odd M\}$
the set of zeros of $T_{M}(x)$. 
This defines the bijection 
$\odd M \to R_M, t \mapsto r_{M}(t)$ 
whose inverse we denote by 
$r \mapsto t_{M}(r)$.

We write $X\eqdef\{\lst pn\}$ and identify 
subsets $\phi\subseteq X$ with  Boolean assignments to the literals $\lst xn$. 
More specifically, $x_{i}$ is assigned "true" if and only if $p_{i}\in\phi$.
We now consider the map 
$$A_{M}:R_{M}  \rightarrow\{0,1\}^{X},\quad
 r  \mapsto\{p_{j}\mid p_{j}\text{ divides }t_{M}(r)\} ,
$$which assigns to a root $r$ of $T_M$ 
the set of prime divisors of $t_M(r)$.
Note that the function $A_{M}$ is surjective but not injective. 
In fact, the fiber of $A_M$ over $\phi\in \{0,1\}^X$ 
is given by 
\begin{equation}\label{eq:SMphieq}
 S_{M}(\phi) \eqdef A_{M}^{-1}(\phi)= \{r\in R_{M}\mid\gcd(t_{M}(r),M)=\alpha(\phi)\} , 
\end{equation} 
where we have set 
\[
\alpha(\phi)\eqdef\prod_{p\in\phi}p .
\]
Finally, we assign to a 3SAT formula $W$ over the literals $\lst xn$ 
the union of the sets~$S_{M}(\phi)$, taken over all satisfying Boolean assignments, that is,  
\[
  S_{M}(W)\eqdef\bigcup_{\phi\in\{0,1\}^{X}\text{, }\phi \text{ satisfies }W}S_{M}(\phi) .
\]
We denote by $\pols_{M}(W)$ the {\em monic} univariate real polynomial 
with the set of roots $S_{M}(W)$:
\[
\pols_{M}(W)\eqdef\prod_{r\in S_{M}(W)}(x-r).
\]
Note that $\pols_{M}(W)$ is square free and only has real roots because it is a factor $T_M(x)$.
As in \cite{PERRUCCI2007471}, we can express $\pols_{M}(W)$ in terms 
of the following analogues of the cyclotomic polynomials: \begin{equation}\label{eq:cyclotomeq}
  C_{\ell}(x)\eqdef\prod_{t\in\odd{\ell}\text{, }\gcd(t,\ell)=1}(x-r_{\ell}(t)).\end{equation}
If $\ell$ is odd, then the degree of $C_{\ell}$ is given by the Euler totient function 
(not to be confused with assignments $\phi$)
\begin{equation}\label{eq:deg-cyclotomeq}
  \deg C_{\ell} = \varphi(\ell) .
\end{equation}
\ref{eq:SMphieq} implies that for an assignment $\psi\in\{0,1\}^{X}$, we
have: 
\begin{equation}\label{eq:cyclotomeq-new}
 \prod_{r\in S_{M}(\psi)}(x-r) = C_{M/\alpha(\psi)}(x) .
\end{equation}
This immediately implies that 
\begin{equation}
\pols_{M}(W)=\prod_{\psi\in\{0,1\}^{X},\text{ }\psi\text{ satisfies }W}C_{M/\alpha(\psi)}(x).\label{eq:eqdefpolysm}
\end{equation}
Note that all the integers $M/\alpha(\psi)$ are odd. Therefore, 
using \ref{eq:deg-cyclotomeq}, we see that 
the number of real roots of $\pols_{M}(W)$ is given by 
\begin{equation}\label{eq:Nroots-PolySAT}
 Z_{\R}(\pols_{M}(W))=\sum_{\psi\text{ satisfies }W}\varphi(M/\alpha(\psi)).
\end{equation}
The following properties of $\pols_{M}(W)$ are easy to verify
(see Lemma~5 in \cite{PERRUCCI2007471}). 
Note that $\frac{1}{2^{M-1}}T_M(x)$ is the monic polynomial obtained by 
dividing $T_M(x)$ by its leading coefficient. 

\begin{lem}\label{lem:polysproperties}Suppose $W,W_{1}$, $W_{2}$ are 3SAT
formulas on the literals $\lst xn$, and the primes $\lst pn$ and $M=\prod_i p_i$ are as before.
Then we have: 
\begin{enumerate}
\item For a literal $x_{i}$, $\pols_{M}(x_{i})=\frac{T_{M/p_{i}}}{2^{M/p_{i}-1}}$. 
\item $W_{1}$ and $W_{2}$ are equivalent iff $\pols_{M}(W_{1})=\pols_{M}(W_{2})$. 
\item $\pols_{M}(\lnot W)=\frac{T_{M}}{2^{M-1}\pols_{M}(W)}.$ 
\item $\pols_{M}(W_{1}\land W_{2})=\gcd(\pols_{M}(W_{1}),\pols_{M}(W_{2}))$. 
\item $\pols_{M}(W_{1}\lor W_{2})=\lcm(\pols_{M}(W_{1}),\pols_{M}(W_{2}))$. 
\end{enumerate}
\end{lem}

This lemma implies that $\pols_{M}(W)$ has rational coefficients. 
We also recall Lemma 6 in \cite{PERRUCCI2007471}, which says that the composition of 
$\pols_{M}(W)$ with the Chebychev polynomial $T_{q}$,
up to a scaling factor, 
equals $\pols_{Mq}(W)$.

\begin{lem}\label{lem:polymtopolymq}Let $W$ be a 3SAT formula over literals
$\lst xn$ and $y$ be a new literal. Let $q$ be a new odd prime associated to the literal $y$. 
If we think of $W$ being a 3SAT formula over the literals $\lst xn,y$, 
then we have, for some $\lambda\in\Q^\ast$, 
\[
 \pols_{M}(W)\circ T_{q} = \lambda \pols_{Mq}(W) .
\]
\end{lem}

We are now concerned with the efficient computation of $\pols_{M}(W)$.
Recall $p_{\max} = \max_i p_i$.

\begin{lem}\label{cor:polysmforaclause}
Let $C$ be a clause formed by 3 literals
$x_{i},x_{j},x_{k}$ (and their negations). In time $O(p_{\max}^{9})$,
we can construct an SLP computing a nonzero integer
polynomial $F_{M}(C)$, such that 
$$ 
 F_{M}(C) = I_M(C)\pols_{M}(C)
$$ 
for some nonzero integer $I_M(C)$. 
\end{lem}

\begin{proof}
Put $N\eqdef p_{i}p_{j}p_{k}$. 
Tracing the proof of Proposition 4 in \cite{PERRUCCI2007471}, we see that 
in time $O(N^{3})$, we can construct an SLP of size $O(N^{3})$ computing
a polynomial $F$, which is a nonzero integer multiple of $\pols_{N}(C)$.
From \ref{lem:polymtopolymq} we deduce that for some $\lambda\in\Q^\ast$, 
$$
 \pols_{N}(C)\circ T_{M/N} = \lambda \pols_{M}(C) .
$$
After multiplying with a suitable integer,
we can write this as  
$F_{M}(C)=I_M(C)\pols_{M}(C)$
with some nonzero integer $I_M(C)$.
With \ref{lem:TpqSLP} we refer that 
$F_M(C)$ has an SLP of size $O(N^{3}+np_{\max})$,
which can be constructed in the same amount of time. To complete the proof, notice that $O(N^{3}+np_{\max})=O(p_{\max}^{9})$. 
\end{proof}

\begin{thm}
\label{thm:pmwsmallcircuit}
Suppose $W=C_{1}\land \dots\land C_{m}$ 
is a 3SAT formula on $n$ literals with $m$ clauses. 
Let $F_{M}(C_{i})$ be the integer polynomial (multiple of $\pols_M(C_i)$) 
constructed for the clause $C_{i}$ in \ref{cor:polysmforaclause}. Then: 
\begin{enumerate}
\item An SLP computing the polynomial $P_{M}(W)$, defined as sum of squares, 
\[
 P_{M}(W)\eqdef\sum_{i\in[m]}(F_{M}(C_{i}))^{2},
\]
can be computed in time $O(mp_{\max}^{9})$ for
given $W$ (and from the primes $p_1,\ldots,p_n$).

\item The polynomial $P_{M}(W)$ has the same set of real roots as $\pols_{M}(W)$.

\item Every real root of $P_{M}(W)$ has multiplicity two. 

\item The radical of $P_{M}(W)$ satisfies 
$$
 \rad(P_{M}(W)) = \pols_{M}(W) \cdot Q_M(W) ,
$$
where $Q_M(W)$ is an integer polynomial having no real roots.
\end{enumerate}
\end{thm}

\begin{proof}
1. The first assertion follows from \ref{cor:polysmforaclause}.

2. A real number $r$ is a root of $P_{M}(W)$ if and only if $r$ is a 
root of all the $F_{M}(C_{i})$. 
The common zero set of 
$P_{M}(C_1),\ldots,P_{M}(C_m)$ 
equals the zero set of $\pols_{M}(W)$, since by \ref{lem:polysproperties} 
\[
\pols_{M}(W)=\gcd(\pols_{M}(C_{1}),\pols_{M}(C_{2}),\dots,\pols_{M}(C_{m})).
\]
Also note that $P_{M}(C_i)$ only has real roots.
It follows that $P_{M}(W)$ and $\pols_{M}(W)$ have same set of real roots, 
which show the second assertion.

3. From the construction of $P_{M}(W)$, it follows that every real root of $\pols_{M}(W)$ has multiplicity two.

4. Let us write 
$\pols_{M}(W)= \prod_{i=1}^a (x-r_i)$.
Then the factorization of $P_M(W)$ into irreducible polynomials in $\Z[x]$ has the form
$$
 P_M(W) = I \cdot \prod_{i=1}^a (x-r_i)^2 \cdot \prod_{j=1}^b h_j^{e_j} ,
$$
where the irreducible polynomials $h_j$ have no real roots, $e_j\ge 1$, and $I\in\Z$.
Therefore
$$
 \rad(P_M(W)) = \rad(I) \prod_{i=1}^a (x-r_i) \cdot \prod_{j=1}^b h_j ,
$$
which is the fourth assertion with $Q_M(W) := \rad(I)\prod_{j=1}^b h_j$.
\end{proof}

Using the above construction, the following was derived in \cite{PERRUCCI2007471}. 
We provide the proof since our argument will be a refinement of it. 

\begin{thm}
\label{thm:nphardnessrealrootexist}\ref{prob:slprealrootexist} is
$\NP$-hard. 
\end{thm}

\begin{proof}
Suppose $W$ is a 3SAT formula on $n$ literals. Using \ref{cor:polysmforaclause} and \ref{thm:pmwsmallcircuit}, 
in time $\poly(p_{\max},m)$, 
we can construct an SLP, which computes a polynomial $f$ 
that has same real roots as $\pols_{M}(W)$, albeit with multiplicity two. 
By definition, $\pols_{M}(W)$
has a real root if and only if $W$ is satisfiable. Hence $f$
has a real root if and only if $W$ is satisfiable. By the well-known
prime number theorem \cite{hardy2008introduction}, $p_{\max}$ can
be chosen to be of magnitude $O(n\log n)$. This proves that \ref{prob:slprealrootexist}
is $\NP$-hard. 
\end{proof}

\section{$\NP$-hardness of PosSLP \label{sec:nphardnessposslp}}
The following problem is crucial for showing the $\NP$-hardness of $\posslp$.

\begin{problem}[Unique-SAT]
\label{prob:uniquesat}
Given a 3SAT formula $W$ with the promise
that $W$ has at most one satisfying assignment, decide if $W$ is
satisfiable. 
\end{problem}

The following is well known.

\begin{thm}[Valiant-Vazirani, \cite{VALIANT198685}]\label{thm:valiantvazirani}
There is a probabilistic polynomial time algorithm, which given a 3SAT formula $W$ on $n$ literals,  
outputs a  3SAT formula $W^{\prime}$ such that if $W$ is satisfiable, 
then $W^{\prime}$ has a unique satisfying assignment with probability at least $\frac{1}{8n}$. 
If $W$ is not satisfiable, then $W^{\prime}$ is also not satisfiable. 
\end{thm}

We now explain how to use~\ref{conj:constrdicalconj} in 
the setting of \ref{thm:pmwsmallcircuit}. This is the key step, 
which makes our result conditional.

\begin{cor}
\label{cor:radicalsmallcircuit}
Suppose $W$ is a 3SAT formula as in \ref{thm:pmwsmallcircuit}. If \ref{conj:constrdicalconj} is true, then in randomized $\poly(p_{\max},m)$ time, we can construct an SLP of size $\poly(p_{\max},m)$, which computes a nonzero multiple $F_{M}(W) = \pols_{M}(W) \cdot Q_M(W)$ of $\pols_{M}(W)$, where the polynomial $Q_M(W)$ does not have real roots. The success probability of this randomized algorithm is at least $1-\frac{1}{\Omega(n^{1+\epsilon})}$.
\end{cor}

\begin{proof}
By using \ref{thm:pmwsmallcircuit}, we know that  $\rad(P_{M}(W)) = \pols_{M}(W) \cdot Q_M(W)$, 
where $Q_W$ is an integer polynomial having no real roots. \ref{thm:pmwsmallcircuit} also shows that $P_{M}(W))$ has an SLP of size $O(mp_{\max}^{9})$ and 
this SLP can also be constructed in time $O(mp_{\max}^{9})$. By using \ref{conj:constrdicalconj}, 
in  randomized $\poly(p_{\max},m)$ time, we can construct an SLP of size $\poly(p_{\max},m)$, 
which computes  $\rad(P_{M}(W))$. By renaming  $\rad(P_{M}(W))$ to $F_{M}(W)$, we obtain the desired claim. 
Since $\tau(P_{M}(W))>n$, we get that the success probability is at least $1-\frac{1}{\Omega(n^{1+\epsilon})}$.
\end{proof}

By \ref{thm:valiantvazirani} we may assume that if a given 3SAT formula $W$ is satisfiable, 
then it has at most one satisfying assignment. Hence we assume in this section that $W$ has exactly one satisfying assignment $\phi$. This assumption implies a simpler structure on the roots of $\pols_{M}(W)$, as seen below in
\ref{lem:uniquesatwpolys}. 
Recall that $\alpha(\phi)$ is defined as $\alpha(\phi)\eqdef\prod_{p_i\in\phi}p_i$.

\begin{lem}\label{lem:uniquesatwpolys}
Suppose a 3SAT formula $W$ over $n$ literals has a unique satisfying assignment $\phi$. Then, writing $N\eqdef\frac{M}{\alpha(\phi)}$, we have 
\begin{equation*}\pols_{M}(W)=C_{N}(x)=\prod_{t\in\odd N\text{, }\gcd(t,N)=1}(x-r_N(t))=\prod_{\gcd(t,2N)=1}(x-r_N(t)).
\end{equation*}
\end{lem}

\begin{proof}
This is an immediate consequence of \ref{eq:cyclotomeq} and \ref{eq:eqdefpolysm}.
\end{proof}

We slightly modify the choice of the primes $\lst p n$ 
in the reduction sketched in \ref{sec:tchebyshevandrealroots} to insure that $p_{\min}\geq n^{3}$.
Moreover, without loss of generality, we may assume that the unique satisfying assignment~$\phi$ 
in \ref{lem:uniquesatwpolys} (if it exists) 
assigns at least one literal $x_{i}$ to be \emph{false}. This is possible since we can first
check if the \emph{all true assignment}, where all literals are true, satisfies $W$ or not.
This entails that $\alpha(\phi)\leq\frac{M}{p_{\min}},$ hence we can assume that
\begin{equation}\label{eq:lowerboundN} 
N\geq p_{\min}\geq n^{3}.
\end{equation}
We think of the multiplicative group $\Z_{2N}^{\times}$ as the subset of $[2N]$ in decreasing order
$$
 \Z_{2N}^{\times}\eqdef\{t\in[2N]\mid \gcd(t,2N)=1\} 
  = \{t_1,\ldots,t_{\varphi(N)}\},
$$
where $t_j > t_ {j+1}$, where $t_1=2N-1$ and $t_{\varphi(N)}= 1$. 
The reason for this choice of indexing is that 
$t\mapsto r_N(t)$ is monotonically decreasing, so that we obtain 
$r_N(t_j) < r_N(t_ {j+1})$, see~\ref{eq:def_r_M}. 

We can then rewrite \ref {lem:uniquesatwpolys} as \[
 \pols_{M}(W) = \prod_{j=1}^{\varphi(N)}\paren{x-r_N(t_j)}.
\]
The roots of $\pols_{M}(W)$ subdivide the interval $(-1,1)$ into a collection 
$\mathcal{I}\eqdef \{I_0, \lst{I}{\varphi(N)}\}$ 
of open intervals.
More precisely, we define $$
 I_j\eqdef \big(r_N(t_j),r_N(t_{j+1})\big) \quad \mbox{for $1\leq j <  \varphi(N)$,}
$$
and $I_0\eqdef(-1, r_N(2N-1))$, $I_{\varphi(N)}\eqdef(r_N(1), 1)$.

Suppose $F_{M}(W)$ is the polynomial constructed in \ref{cor:radicalsmallcircuit}. Since the real roots  of $F_{M}(W)$ are exactly that of $\pols_{M}(W)$, and all these 
roots are of multiplicity one, $F_{M}(W)$ does not change sign on the interval $I_0$.
Without loss of generality, we assume that $F_{M}(W)$ is positive on the leftmost interval $I_{0}$. 
Using \ref{lem:rootcountparity} we then infer that

\centerline{$F_{M}(W)$ is negative on the interval $I_j\in\mathcal{I}$ 
if and only if $j$ is odd.}
From now on, to simplify notation, we drop the index $M$ and argument $W$ in this  section
and just write $F:=F_{M}(W)$. 
The above discussion implies that if we pick any real number $a$ from any interval~$I_j$ in $\mathcal{I}$
of odd index~$j$,
then $F(a)$ and $F(1)$ have different signs\ie
$F(a)$ is negative. An interval $I_j\in\mathcal{I}$ is said to be odd-indexed if $j$ is odd, otherwise it is said to be even-indexed. Now our overall strategy can be summarized as follows
(for a formal argument see \ref{lem:intermedred}):
\begin{enumerate}
\item Show that the sum of the lengths of odd-indexed intervals $I_j\in\mathcal{I}$ 
 is at least $c$ for some positive constant $c$. 
\item Pick a ``random'' rational number $a$ in the interval $(-1,1)$.
  With probability at least $c/2$, we have $a\in I_j$  with $j$ being odd.
\item Compute the sign of $F(a)F(1)$ using $\posslp$. If $F(a)F(1) < 0$
then $\pols_{M}(W)$ has a real root and hence $W$ is satisfiable.
This succeeds with probability at least $c/2$ if $W$ is
satisfiable. 
\end{enumerate}

 Out next goal is to show the following result.

\begin{prop}\label{thm:total_len_good_interval}If $W$ is satisfiable, then
the sum of the lengths of odd-indexed intervals in $\mathcal{I}$ is at least $\frac{1}{\pi}$. 
\end{prop}

For this, we rely on the two technical lemmas below, whose proof is 
postponed to \ref{se:technical}. We call a subinterval $J$ of $[-1,1]$ 
a \emph{{\twodiff} interval} if it connects two subsequent real roots of the Chebychev polynomial $T_{N}$. 
This means that  $J=((r_N(s_1),r_N(s_2))$ for $s_{1},s_{2}\in\odd N$ with $s_{1}-s_{2}=2$. 
Note that the end points of a simple such interval are not required to be zeros of $\pols_{M}(W)$.

The first lemma states that a substantial part of the interval $(-1,1)$ 
is subsumed by {\twodiff} intervals in $\mathcal{I}$.

\begin{lem}\label{lem:nontwodiffsmall}
The sum of the lengths of non-{\twodiff} intervals in $\mathcal{I}$
is at most $\frac{1}{n}$.
\end{lem}

We also need that two adjacent {\twodiff} intervals cannot differ too much in their lengths. 

\begin{lem}
\label{lem:ratioofconsectivegreenred}
Suppose $I,J$  are two adjacent {\twodiff} intervals (not necessarily in $\mathcal{I}$). 
Then we have: 
\[
\frac{1}{\pi}\leq\frac{\len(I)}{\len(J)}\leq\pi.
\]
\end{lem}

\begin{proof}[Proof of \ref{thm:total_len_good_interval}]
We denote by $\ell_{se}$ the sum of the lengths of {\twodiff} even-indexed intervals in $\mathcal{I}$. 
Analogously, $\ell_{so}$ denotes the sum of the lengths of {\twodiff} odd-indexed intervals in $\mathcal{I}$,
and we write $\ell_{o}$ for the sum of the lengths of \emph{all} odd-indexed intervals in $\mathcal{I}$. 
By \ref{lem:nontwodiffsmall} we have 
$ 2 \le \ell_{se} + \ell_{so} +  \frac{1}{n} \le \ell_{se} + \ell_{o} +  \frac{1}{n}$, 
hence $\ell_{se} \geq 2 - \ell_{o} - \frac{1}{n}$. 

Consider a {\twodiff} even-indexed interval $I_{j}\in\mathcal{I}$. 
The interval $I_{j+1}$ contains a {\twodiff} interval~$J$. 
Therefore $\abs{I_{j+1}}/\abs{I_{j}}\geq \abs{J}/\abs{I_{j}}\ge \pi^{-1}$, 
where the right inequality is due to~\ref{lem:ratioofconsectivegreenred}.
Applying this argument to all {\twodiff} even-indexed intervals in $\mathcal{I}$,  
and adding their length, we obtain that: 
$$
 \ell_{o} \ge  \frac{1}{\pi} \ell_{se} \geq \frac{1}{\pi}\big(2 - \ell_{o} - \frac{1}{n}\big).
 $$ 
 A simple calculation shows that the above equation implies 
 $\ell_{o} \geq \frac{1}{\pi+1}\paren{2 - \frac{1}{n}}\geq \frac{1}{\pi}$ for $n>1$.
\end{proof}

\begin{lem}\label{lem:intermedred}
Assume $M\ge 3$. Suppose $K$ is an integer in $[-M^{4},M^{4}]\cap\Z$ chosen uniformly at random.
If $W$ is satisfiable, then we have $\mathrm{Prob}_K[F(K/M^{4}) F(1)<0] \geq \frac{1}{4\pi}$.
On the other hand, we always have $F(K/M^{4}) F(1)>0$ if $W$ is not satisfiable. 
\end{lem}

\begin{proof}
If suffices to consider the case where $W$ is satisfiable.
We can think of sampling a point $a$ in the set of grid points
$\Gamma := [-1,1] \cap \frac{1}{M^4}\Z$ , which has cardinality $2M^4+1$.
Recall that $F(a) F(1) < 0$ implies that $a$ lies in an odd-index interval $I_j\in\mathcal{I}$.

We claim that for $I_j\in\mathcal{I}$ with odd index $j$
$$
 \card{I_j \cap \Gamma} \ge M^4 \len(I_j) -2 .
$$
Indeed, assume $I_j=(a,b)$ and 
let $a^{\prime}\ge a$ be minimal such that $a^{\prime} M^4\in\Z$. 
Similarly, let $b^{\prime} \le b$ be maximal such that $b^{\prime}M^4\in\Z$. 
Then $a^{\prime}-a \leq M^{-4}$ and $b-b^{\prime}\leq M^{-4}$. 
We obtain 
\[
\card{I_j \cap \Gamma} = (b^{\prime}-a^{\prime})M^{4}+1\ge (b-a)M^{4}- 2 , \]
which shows the claim. 

Summing over all odd-indexed intervals $I_j$ in $\mathcal{I}$, we get 
$$
 \sum_{j} \card{I_j \cap \Gamma} \ge M^4\, \sum_{j} \len(I_j) - 2 M \ge \frac{M^4}{\pi} - 2M ,
$$
where we used \ref{thm:total_len_good_interval} for the right-hand inequality
(here we use that $W$ is satisfiable). 
 
We can lower bound the probability that 
a uniformly random point $a\in \Gamma$ lands in an odd-indexed interval $I_j\in \mathcal{I}$ as follows (let $M\ge 3$) :
$$
 \frac{1}{\card{\Gamma}} \sum_{j} \card{I_j \cap \Gamma} \ge \frac{1}{2M^4+1} \Big(\frac{M^4}{\pi} - 2M\Big)
   \ge \frac{1}{ 4\pi} , 
$$
where the right-hand inequality follows from a simple calculation. 
This completes the proof
\end{proof}

We now restate the main result of this paper.\mainresultnphardposslp*

\begin{proof}
Given a 3SAT formula $W$, we use \ref{thm:valiantvazirani} to compute a 3SAT formula $W^{\prime}$ such that if $W$ is satisfiable, 
then $W^{\prime}$ has a unique satisfying assignment with probability at least $\frac{1}{8n}$. 
If $W$ is not satisfiable then $W^{\prime}$ is also not satisfiable. 
Now we use \ref{cor:radicalsmallcircuit} to compute an SLP which computes $F=F_M(W^{\prime})=\rad(P_M(W^{\prime}))$. 
Then we randomly sample $K\in [-M^{4},M^{4}]\cap\Z$ 
and compute the sign of $(-1)F(K/M^4) F(1)$ based on \ref{lem:signatrational}.

By using the assumption $\posslp\in\BPP$ and amplifying the probability using
standard probability amplification techniques,
we assume that success probability of a $\posslp$ oracle is at least $1-2^{-n}$, 
see \cite{complexityarora2009} or more specifically \cite[Deﬁnition 1]{sudhanbpplec2007}.

Using~\ref{lem:intermedred} on $F=F_M(W^{\prime})$, we see that 
if $W$~is satisfiable, then $(-1)F(K/M^4) F(1) > 0$ happens with probability at least $\frac{1}{4\pi}\cdot\frac{1}{8n}\cdot\paren{1-\frac{1}{\Omega(n^{1+\epsilon})}}=\Omega\paren{\frac{1}{n}}$. 
Hence the $\posslp$ oracle verifies the inequality $(-1)F(K/M^4) F(1) > 0$  with probability at least 
$\Omega\paren{\frac{1}{n}}\cdot\paren{1-2^{-n}}=\Omega\paren{\frac{1}{n}}$. 

On the other hand  if $W$ is not satisfiable, then neither is $W^{\prime}$.
Hence $(-1)F(K/M^4) F(1)$ can only be positive if  $F\neq\rad(P_M(W^{\prime}))$,
\ref{cor:radicalsmallcircuit}  implies that this happens with probability at most  $O\paren{\frac{1}{n^{1+\epsilon}}}$.
Hence $(-1)F(K/M^4) F(1) > 0$ happens with probability at most $O\paren{\frac{1}{n^{1+\epsilon}}}$.
Hence the $\posslp$ oracle verifies the inequality $(-1)F(K/M^4) F(1) > 0$   with probability at most 
$O\paren{\frac{1}{n^{1+\epsilon}}} + 2^{-n}=O\paren{\frac{1}{n^{1+\epsilon}}}$.
After doing $\poly(n)$ many independent runs of this reduction, we can boost the success probability for both cases to $O(1)$
using standard probability amplification techniques, see \cite{complexityarora2009} or more specifically \cite[Theorem 3]{sudhanbpplec2007}.
Altogether, this shows that the 3SAT problem lies in $\BPP$. 
This implies that $\NP \subseteq \BPP$.

\end{proof}

\subsection{Proofs of Technical Lemmas}
\label{se:technical}

\begin{proof}[Proof of \ref{lem:nontwodiffsmall}]
There are two kind of  non-{\twodiff} intervals in $\mathcal{I}$. The first kind are the leftmost 
and the rightmost intervals: 
namely $I_0=\left(-1,\cos\left(\frac{(2N-1)\pi} {2N}\right)\right)$ and 
$I_{\varphi(N)}=\left(\cos\left(\frac{\pi}{2N}\right),1\right)$. 
We use the
inequality $1-\cos(x)\leq\frac{x^{2}}{2}$ to infer that the length
of the interval $I_0$
is at most $\frac{\pi^{2}}{8N^{2}}$. Using  a similar argument, the length of interval $I_{\varphi(N)}$ can also be upper bounded by $\frac{\pi^{2}}{8N^{2}}$.
Hence the total length of both these intervals is at most $\frac{\pi^{2}}{4N^{2}}$. 

The second kind of  non-{\twodiff} intervals in $\mathcal{I}$ are of the form:
\begin{equation}\label{eq:nonsimpleintervalsecondind}
(r_N(s_1)),r_N(s_2))=\left(\cos\left(s_{1}\frac{\pi}{2N}\right),\cos\left(s_{2}\frac{\pi}{2N}\right)\right) \text{with } s_1,s_2\in\odd{N}  \text{ and } (s_{1}-s_{2})\geq 4. 
\end{equation}
Now we bound the length of these second kind of  non-{\twodiff} intervals.
To this end we define: $D\eqdef\{r_N(t)\mid t\in\odd N\setminus\Z_{2N}^{\times}\}$. The elements of $D$ are  the roots of $T_N$ which are not the roots of $\pols_M(W)=C_N$, since the root set of $C_N$ is exactly $\{r_N(t)\mid t\in\Z_{2N}^{\times}\}$. We have $\card D = N-\varphi(N)$. Suppose $I=((r_N(s_1)),r_N(s_2))\in\mathcal{I}$ is a  non-{\twodiff} interval as in \ref{eq:nonsimpleintervalsecondind}. Since $I$ is non-{\twodiff}, its endpoints are not subsequent roots of $T_N$. Hence $I$ contains some roots of $T_N$ which are not the roots of $C_N$. More precisely, there are exactly $\frac{s_{1}-s_{2}}{2}-1$ elements of $D$ in $I$.
The length of such an interval $I$ can be bounded as: 
\[
\len(I)=\cos\left(s_{2}\frac{\pi}{2N}\right)-\cos\left(s_{1}\frac{\pi}{2N}\right)=2\sin\left(\frac{(s_{1}+s_{2})\pi}{4N}\right)\sin\left(\frac{(s_{1}-s_{2})\pi}{4N}\right)\leq\frac{(s_{1}-s_{2})\pi}{2N}.
\]
 Since $(s_{1}-s_{2})\geq 4$ implies that $(s_{1}-s_{2})\le 2(s_{1}-s_{2})-4$ , we get that
the length of such an $I$ can be bounded as: 
\[
\len(I)\leq\frac{(s_{1}-s_{2})\pi}{2N}\leq\frac{2\pi}{N}\left(\frac{s_{1}-s_{2}}{2}-1\right).
\]
Hence if a non-{\twodiff} interval $I$ contains $m$
elements of $D$, then we have: 
\[
\len(I)\leq\frac{2\pi}{N}m.
\]
Since $\card D = N-\varphi(N)$, we get that the total length of such non-{\twodiff}
intervals is at most $\frac{2\pi(N-\varphi(N))}{N}$. This proves that the total length of non-{\twodiff}
intervals is at most $\frac{2\pi(N-\varphi(N))}{N}+\frac{\pi^{2}}{4N^{2}}$. Now we use the inequality:
\begin{equation}
1-x\geq e^{-2x}\text{ for all }0\leq x\leq\frac{1}{2}\label{eq:e2xineq}
\end{equation}
to first obtain that:
\begin{align}
\frac{\varphi(N)}{N} & =\prod_{p_{i}\mid N}\left(1-\frac{1}{p_{i}}\right)\geq\left(1-\frac{1}{p_{\min}}\right)^{n}\nonumber \\
 & \geq e^{-\frac{2n}{p_{\min}}}\geq e^{-\frac{2}{n^{2}}}\geq1-\frac{2}{n^{2}}.\label{eq:phiNoverNlowerbound}
\end{align}
By using \ref{eq:phiNoverNlowerbound}, we obtain:
\begin{align*}
\frac{2\pi(N-\varphi(N))}{N} +\frac{\pi^{2}}{4N^{2}} & \leq2\pi\left(1-\frac{\varphi(N)}{N}\right)+\frac{\pi^{2}}{4N^{2}}\leq2\pi\frac{2}{n^{2}}+\frac{\pi^{2}}{4N^{2}}\\
 & \leq\frac{4\pi}{n^{2}}+\frac{\pi^{2}}{4n^{6}}\leq\frac{1}{n}.
\end{align*}
In proving the above upper bound, we have have used the lower bound $N\geq p_{\min} \geq n^{3}$, as established in \ref{eq:lowerboundN}. We also assumed $n$ to be large enough.
Hence the sum of the lengths of the non-{\twodiff} intervals is at most~$\frac{1}{n}$.
\end{proof}

\begin{proof}[Proof of \ref{lem:ratioofconsectivegreenred}]
Let $I:=\left(\cos\left(t\frac{\pi}{2N}\right),\cos\left((t+2)\frac{\pi}{2N}\right)\right)$,
$J:=\left(\cos\left((t+2)\frac{\pi}{2N}\right),\cos\left((t+4)\frac{\pi}{2N}\right)\right)$.
We use the trigonometric identity $\cos A-\cos B=2\sin\left(\frac{A+B}{2}\right)\sin\left(\frac{B-A}{2}\right)$
to obtain: 
\begin{alignat*}{1}
\frac{\cos\left(t\frac{\pi}{2N}\right)-\cos\left((t+2)\frac{\pi}{2N}\right)}{\cos\left((t+2)\frac{\pi}{2N}\right)-\cos\left((t+4)\frac{\pi}{2N}\right)} & =\frac{\sin\left((t+1)\frac{\pi}{2N}\right)}{\sin\left((t+3)\frac{\pi}{2N}\right)}.
\end{alignat*}
We have the following cases: 
\begin{casenv}
\item In this case, we assume $(t+3)\frac{\pi}{2N}\leq\frac{\pi}{2}$. Now
we use the well known inequality: 
\[
\frac{2}{\pi}x\leq\sin(x)\leq x\text{ for }0<x\leq\frac{\pi}{2}.
\]
This implies that: 
\[
\frac{2}{\pi}\frac{t+1}{t+3}\leq\frac{\sin\left((t+1)\frac{\pi}{2N}\right)}{\sin\left((t+3)\frac{\pi}{2N}\right)}\leq1.
\]
Since $\frac{t+1}{t+3}\geq\frac{1}{2}$. We obtain that: 
\[
\frac{1}{\pi}\leq\frac{\sin\left((t+1)\frac{\pi}{2N}\right)}{\sin\left((t+3)\frac{\pi}{2N}\right)}\leq1.
\]
\item Now consider the case when $(t+3)\frac{\pi}{2N}>\frac{\pi}{2}$. By
using the equality $\sin(\pi-\theta)=\sin(\theta)$, we obtain: 
\[
\frac{\sin\left((t+1)\frac{\pi}{2N}\right)}{\sin\left((t+3)\frac{\pi}{2N}\right)}=\frac{\sin\left((2N-(t+1))\frac{\pi}{2N}\right)}{\sin\left((2N-(t+3))\frac{\pi}{2N}\right)}.
\]
In this case, we also know that $(2N-(t+1))\frac{\pi}{2N}\leq\frac{\pi}{2}$.
By using the result of the first case, we know that: 
\[
\frac{1}{\pi}\leq\frac{\sin\left((2N-(t+3))\frac{\pi}{2N}\right)}{\sin\left((2N-(t+1))\frac{\pi}{2N}\right)}\leq1 ,
\]
which proves the claim. 
\end{casenv}
\end{proof}

\section{Hardness of Counting Real Roots}\label{sec:realrootcountinghardness}

We first recall some standard definitions of counting complexity classes from \cite{complexityarora2009}. 
A function $f:\{0,1\}^{*}\to\N$ is in $\#\P$ if there exists a
polynomial $p:\N\to\N$ and a polynomial-time Turing machine $M$
such that for every $x\in\{0,1\}^{*}$ 
\[
  f(x)=\card{\{y\in\{0,1\}^{p\paren{\card x}}:M(x,y)=1\}}.
\]
A function $f:\{0,1\}^{*}\to\N$ is $\#\P$-hard if every $g\in\#\P$
can be computed in polynomial time, allowing oracle calls to $f$.
We say that $f$ is $\#\P$-complete if it is in $\#\P$ and $\#\P$-hard.
We denote by \#3SAT the problem of computing, for a given 3SAT formula~$W$, 
the number of satisfying assignments for~$W$. 
It is well known that \#3SAT is $\#\P$-complete.

The main result of this section can be restated as: 
\countrealrootssharppahrd*

As in~\ref{subsec:RR},
for a given 3SAT formula $W$ defined over $n$ literals, 
we choose $n$ distinct odd primes $\lst pn$
associated with the literals, and we set $M\eqdef\prod_{i\in[n]}p_{i}$.
According to \ref{thm:pmwsmallcircuit}, for given $W$, 
one can compute in polynomial time an SLP, which computes a univariate integral polynomial
$P_{M}(W)$, which has the same set of real roots 
as $\pols_{M}(W)$, albeit with multiplicity two.
From~\ref{eq:Nroots-PolySAT}, we obtain the following equality for the number of real roots.
\begin{equation}\label{eq:ZRPM}
 \frac12 Z_{\R}(P_{M}(W)) = Z_{\R}(\pols_{M}(W)) = \sum_{\psi\text{ satisfies }W}\varphi(M/\alpha(\psi)).
\end{equation}

\begin{lem}
\label{lem:satsmodq}
Suppose $q$ is an odd prime. Let $\lst pn$
be  $n$ distinct primes in the arithmetic progression $\{aq+2\mid a\in\N_{>0}\}$.
As above, define $M\eqdef\prod_{i\in[n]}p_{i}$. Then for any 3SAT
formula $W$ defined over $n$ literals, we have: 
\[
\#W\bmod q \equiv \frac{1}{2}Z_{\R}(P_{M}(W))\bmod q.
\]
\end{lem}

\begin{proof}
By assumption, $p_{i}$ is an odd prime and 
$p_{i}-1\equiv 1 \bmod q \label{eq:pminus1modq}$.
Let $N=M/\alpha(\psi)$ be as in \ref{eq:ZRPM}, say
$N=\prod_{i\in I}p_{i}$ for $I\subseteq [n]$. 
Then,
\[
  \varphi(N)\bmod q \equiv\prod_{i\in I}(p_{i}-1)\equiv 1 \bmod q .
\]
Therefore, using \ref{eq:ZRPM}, \[
 Z_{\R}(\pols_{M}(W)) \equiv\sum_{\psi\text{ satisfies }W}1\bmod q\equiv\#W\bmod q,
 \]
 which shows the assertion.
 \end{proof}

We also need the following lemma. 

\begin{lem}
\label{lem:primesmodqi}
There is a polynomial time algorithm, which on input a natural number $n$ (encoded in unary) computes distinct odd primes $\lst qn$ and, for each $i\in[n]$, 
computes a collection $p_{i1},\ldots, p_{in}$ of distinct primes 
such that $p_{ij}\equiv2\bmod q_{i}$ for all \textup{$i\in[n],j\in[n]$}. 
\end{lem}

\begin{proof}
We can directly use the algorithm developed for Theorem~4.11 of \cite{countingcurves1993}.
Step 3 of that algorithm can be skipped.
Moreover, after replacing the randomized primality testing by 
deterministic primality testing of~\cite{aksprimesinp2004}, the algorithm becomes deterministic polynomial time.
\end{proof}

\begin{proof}[Proof of \ref{thm:countrealrootshard}]
We reduce \#3SAT to $\crr$. For given $n$,
we first compute in polynomial time the primes $q_i$ and $p_{ij}$ for $i,j\in [n]$
as in~\ref{lem:primesmodqi}.
Then we set $M_i\eqdef p_{i1}p_{i2}\ldots p_{in}$ for $i\in [n]$.
For a given a 3SAT formula $W$ over $n$ literals, 
we first compute in polynomial time SLPs computing the polynomials $P_{M_i}(W)$ for $i\in [n]$ 
according to \ref{thm:pmwsmallcircuit}. 
Then we compute the number of real zeros $Z_{\R}(P_{M_i}(W))$ for $i\in [n]$ by oracle calls to $\crr$.
We divide these even numbers by two and have thus computed
$\#W\bmod q_i$ for $i\in [n]$ according to \ref{lem:satsmodq}.
Since $\#W\leq2^{n}$ and $\prod q_{i}>2^{n+1}$, by using efficient algorithms for Chinese
remaindering \cite{MCA13}, we can recover $\#W$ in polynomial time. 
We have thus shown that $\crr$ is $\#\P$-hard.
\end{proof}

\section{Complexity of Radicals\label{sec:compelxityofradicals}}

\subsection{Complexity of Factors}

For a polynomial $f\in\F[\lst xn]$, we define $L(f)$ as the size
of the smallest arithmetic circuit computing $f$ 
from the variables $x_i$ and {\em any constants} in the field $\F$.

The {\em Factor Conjecture}~\cite{complredpeterbook2000,compl-factors-jfocm04}
is a central question in algebraic complexity theory. 
It asks whether complexity of factors $g$ of $f$ can be bounded by a polynomial 
in terms of $L(f)$ and the degree of $g$.

\begin{conjecture}[Factor Conjecture]\label{conj:factorconjecture} 
Over a field $\F$ of characteristic zero, 
we have $L(g)\leq\poly(L(f)+\deg(g))$ for any factor $g$ of $f\in\F[\lst xn]$.
\end{conjecture}

Evidence for this conjecture comes from~\cite{compl-factors-jfocm04},
where it was shown that the conjecture is true for approximate complexity. 
In addition, the factor conjecture is supported by the following result. 
It tells us that the only situation in which \ref{conj:factorconjecture} could fail 
is when the factor $g$ occurs with exponentially large multiplicity~$e$. 

\begin{thm}[\cite{kaltoffenfactoringhensel1987,complredpeterbook2000}]
\label{thm:kaltoffenfactoring}
If $f,g,h\in\F[\lst xn]$ are such
that $f=g^{e}h$ with $g,h$ being coprime, the multiplicity $e$ a positive integer, 
and $\chr(\F)=0$, then, setting $d=\deg g$, we have 
\[
L(g)=O(M(d)^{2}M(de)(L(f)+n+d\log e)).
\]
Here $M(m) = O(m^2)$ stands for the number of arithmetic operations sufficient to multiply two degree $m$ polynomials over $\F$.
\end{thm}

\begin{comment}
We showed in \ref{thm:pmwsmallcircuit} that we can construct a small
SLP for a polynomial $P_{M}(W)$ which has same set of real roots
as that of $\pols_{M}(W)$, albeit with multiplicity two. By definition
of $\pols_{M}(W)$, $\pols_{M}(W)$ has a real root if and only if
$W$ is satisfiable. Hence $P_{M}(W)$ has a real root if and only
if $W$ is satisfiable. This immediately implies that \ref{prob:slprealrootexist}
is $\NP$-hard. To extend this idea to $\NP$-hardness of $\posslp$,
\end{comment}

Now consider the following related conjecture proposed in
\cite[Conjecture 1]{dssjacm22}. 

\begin{conjecture}[Radical conjecture]
\label{conj:radsaxenaconjecture}
For a nonzero polynomial $f\in\Z[\lst xn]$ we have 
\[
\min\{\deg(\rad(f)),L(\rad(f))\}\le\poly(L(f)).
\]
\end{conjecture}

This essentially states that either the
degree of $\rad(f)$ or the complexity of $\rad(f)$
(or both), are polynomially bounded in the complexity of $f$.
For a univariate polynomial~$f\in\Z[x]$, 
\ref{conj:radsaxenaconjecture} implies that 
\begin{equation}\label{conj:univar-constr}
     L(\rad(f))\leq\poly(L(f)).
\end{equation}
This follows since $L(g)=O(\deg(g))$ for any $g\in\Z[x]$.

To arrive at the main contribution of our paper (\ref{thm:mainnphardposslp}), 
we rely on the following constructive variant of 
\ref{conj:univar-constr}, formulated for $\tau(f)$ instead
of $L(f)$. (Recall that $\tau(f)$ denotes the complexity for constant-free arithmetic circuits.) 

\constructiveradicalconjecture*

\begin{remark}
The main difference of these radical conjectures, in comparison with the factor conjecture, 
is that they make a statement about the complexity of factors $g$ having huge degree. 
While there is currently not much evidence for these radical conjectures, one may still view 
\ref{thm:kaltoffenfactoring} as an indication towards them, since it shows that 
small multiplicities help. When passing to the radical to $f$, all multiplicities
of the irreducible factors $g$ of $f$ are set to one. As for the plausibility of the constructive version 
\ref{conj:constrdicalconj}, we just note that statements on the existence of small 
SLPs usually can be refined to statements concerning randomized effective constructions.
\end{remark}

\section{Conclusion and Open Questions}

Assuming the constructive radical conjecture (\ref{conj:constrdicalconj}),
we proved that $\posslp\in\BPP$ would imply $\NP \subseteq \BPP$. This was achieved by first reducing 3SAT to Unique-SAT \cite{VALIANT198685}
and then using ideas developed in \cite{PLAISTED1984125, PERRUCCI2007471}. This
leads to the first non-trivial lower bound for $\posslp$, albeit
conditional. Using the ideas developed in \cite{countingcurves1993},
we also proved that counting the real roots of univariate polynomials
computed by a given SLP is $\#\P$-hard. 

There are several avenues for further research:
\begin{enumerate}
\item Of course it remains intriguing to prove the unconditional $\NP$-hardness
of $\posslp$. In our approach, we first constructed an SLP for the 
polynomial $P_{M}(W)$, which has same set of real roots as $\pols_{M}(W)$,
but with multiplicity two. 
One could try to directly construct a polynomial size SLP for $\pols_{M}(W)$, which has simple
real roots. But this is unlikely because $\pols_{M}(W)=1$ iff $W$
is not satisfiable. And we can check $\pols_{M}(W)=1$ by polynomial
identify testing. This would imply $\coNP\subseteq\coRP$, which is
not believed to be true.

\item We saw that the sum-of-square-roots problem and inequality testing of
succinctly represented integers are special cases of $\posslp$.
We do not know of any non-trivial upper or lower bounds for these problems. Perhaps one can first study the complexity of these special cases.

\item In addition, the following special case of $\posslp$ are worth investigating:
A univariate polynomial $f(x)=a+bx^{\beta}+cx^{\gamma}\in\Z[x]$ with
$\beta,\gamma\in\N$ is called a {\em trinomial}. Koiran \cite{koirantrinomial2019}
proved that the roots of trinomials are ``well-separated'' and 
he posed the following problem: Given a rational $p/q$ and a trinomial $f(x)$ as inputs, determine the sign of $f(p/q)$. This problem is easily seen to be a special case of $\posslp$. 
\cite{Boniface2022TrinomialsAD} proved that this
problem can be solved in deterministic polynomial time, except on a $\frac{1}{\Omega(\log(dH))}$
fraction of the inputs. Here $\deg(f)\leq d$ and $\max\{\abs a,\abs b,\abs c,\abs p,\abs q\}\leq H$.
Can we find efficient algorithms for Koiran's problem?

\item \cite{PLAISTED1984125} proved that the following problem is $\NP$-hard: 
  decide for a given $k$-sparse polynomial whether it has a root on the unit circle. 
  Can we prove a similar lower bound for  $k$-{\em sparse polynomials}? 
\end{enumerate}
\newcommand{\etalchar}[1]{$^{#1}$}

 \end{document}